\documentclass[11pt]{article}

\usepackage{hyperref}

\usepackage[margin = 3cm]{geometry}
\usepackage{amsthm}
\usepackage{amsmath}
\usepackage{natbib}
\usepackage{graphicx}

\usepackage{amsfonts}
\usepackage{amscd}
\usepackage{amssymb}
\usepackage{mathtools}
\usepackage{float}
\usepackage{subcaption}
\usepackage{enumerate}

\usepackage{algorithm}
\usepackage{algorithmic}

\numberwithin{equation}{section}

\newcommand{\vx}{\mathbf{x}}

\newcommand{\vbeta}{\boldsymbol{\beta}}
\newcommand{\vgamma}{\boldsymbol{\gamma}}
\newcommand{\vmu}{\boldsymbol{\mu}}

\newcommand{\vepsilon}{\boldsymbol{\varepsilon}}

\newcommand{\vzero}{\mathbf{0}}

\newcommand{\mH}{\mathbf{H}}
\newcommand{\mI}{\mathbf{I}}

\newcommand{\mV}{\mathbf{V}}

\newcommand{\mX}{\mathbf{X}}
\newcommand{\mY}{\mathbf{Y}}

\newcommand{\mSigma}{\boldsymbol{\Sigma}}

\DeclareMathOperator*{\argmin}{arg\,min}
\DeclareMathOperator*{\argmax}{arg\,max}

\newcommand{\E}{\mathbb{E}}

\newcommand{\R}{\mathbb{R}}

\theoremstyle{plain}
\newtheorem{theorem}{Theorem}

\newtheorem{proposition}[theorem]{Proposition}

\begin{document}

\title{Variance prior forms for high-dimensional Bayesian variable selection}

\author{Gemma E. Moran \thanks{Department of Statistics, The Wharton School, University of Pennsylvania, Philadelphia, PA 19104. Email: \texttt{gmoran@wharton.upenn.edu}}, Veronika Ro\v{c}kov\'{a} \thanks{Booth School of Business, University of Chicago, Chicago, IL 60637. }, and Edward I. George \thanks{Department of Statistics, The Wharton School, University of Pennsylvania, Philadelphia, PA 19104.  
}}

\maketitle

\begin{abstract}
Consider the problem of high dimensional variable selection for the Gaussian linear model when the unknown error variance is also of interest. In this paper, we show that the use of conjugate shrinkage priors for Bayesian variable selection can have detrimental consequences for such variance estimation.  Such priors are often motivated by the invariance argument of \citet{J61}.  Revisiting this work, however, we highlight a caveat that Jeffreys himself noticed; namely that biased estimators can result from inducing dependence between parameters \emph{a priori}. In a similar way, we show that conjugate priors for linear regression, which induce prior dependence, can lead to such underestimation in the Bayesian high-dimensional regression setting.  Following Jeffreys, we recommend as a remedy to treat regression coefficients and the error variance as independent \emph{a priori}. Using such an independence prior framework, we extend the Spike-and-Slab Lasso of \citet{RG18} to the unknown variance case. This extended procedure outperforms both the fixed variance approach and alternative penalized likelihood methods on simulated data.  On the protein activity dataset of \citet{CP98}, the Spike-and-Slab Lasso with unknown variance achieves lower cross-validation error than alternative penalized likelihood methods, demonstrating the gains in predictive accuracy afforded by simultaneous error variance estimation. 
\end{abstract}

\section{Introduction}

Consider the classical linear regression model
\begin{equation}
\mY= \mX\vbeta + \vepsilon, \quad \vepsilon \sim N_n(0, \sigma^2\mI_n) \label{regression}
\end{equation}
where $\mY\in\R^n$ is a vector of responses, $\mX = [\mX_1,\dots, \mX_p] \in \R^{n\times p}$ is a fixed regression matrix of $p$ potential predictors, $\vbeta = (\beta_1,\dots, \beta_p)^T \in \R^p$ is a vector of unknown regression coefficients and $\vepsilon \in \R^n$ is the noise vector of independent normal random variables with $\sigma^2$ as their unknown common variance. 

When $\vbeta$ is sparse so that most of its elements are zero or negligible, finding the non-negligible elements of $\vbeta$, the so-called variable selection problem, is of particular importance.   Whilst this problem has been studied extensively from both frequentist and Bayesian perspectives, much less attention has been given to {the simultaneous estimation of the error variance $\sigma^2$. Accurate estimates of $\sigma^2$} are important to discourage fitting the noise {beyond the the signal, thereby helping to} mitigate overfitting of the data. Variance estimation is also essential in uncertainty quantification for inference and prediction.

In the frequentist literature, the question of estimating the error variance in our setting has begun to be addressed with papers including the scaled Lasso \citep{SZ12} and the square-root Lasso \citep{B14}.  Contrastingly, in the Bayesian literature, the error variance has been fairly straightforwardly estimated by including $\sigma^2$ in  prior specifications. Despite this conceptual simplicity, the majority of theoretical guarantees for Bayesian procedures restrict attention to the case of known $\sigma^2$, as there is not a generally agreed upon prior specification when $\sigma^2$ is unknown. More specifically, priors on $\vbeta$ and $\sigma^2$ are typically introduced in one of two ways: either via a conjugate prior framework or via an independence prior framework.

Conjugate priors have played a major role in regression analyses. The conjugate prior framework for \eqref{regression} begins with specifying a prior on $\vbeta$ that depends on $\sigma^2$ as follows:
\begin{align}
\vbeta |\sigma^2 \sim N(0, \sigma^2\mV), \label{conj}
\end{align}
where $\mV$ may be fixed or random. This prior \eqref{conj} results in a Gaussian posterior for $\vbeta$ and as such is conjugate. To complete the framework, $\sigma^2$ is assigned an inverse-gamma (or equivalently scaled-inverse-$\chi^2$) prior. A common choice in this regard is the right-Haar prior for the location-scale group \citep{B98}:
\begin{align}
\pi(\sigma) \propto 1/\sigma. \label{right_haar}
\end{align}
Whilst the right-Haar prior is improper, it can be viewed as the limit of an inverse-gamma density. When combined with \eqref{conj}, the prior \eqref{right_haar} results in an inverse-gamma posterior for $\sigma^2$ and as such it behaves as a conjugate prior. 

Prominent examples that utilize the above conjugate prior framework include:
\begin{itemize}
\item Bayesian ridge regression priors, with $\mV = \tau^2\mI$;
\item Zellner's $g$-prior, with $\mV = g(\mX^T\mX)^{-1}$; and
\item Gaussian global-local shrinkage priors, with $\mV = \tau^2\Lambda,$ for $\Lambda = \text{diag}\{\lambda_j\}_{j=1}^p.$
\end{itemize}

We note that the conjugate prior framework refers only to the prior characterization of $\vbeta$ and $\sigma^2$, and allows for any prior specification on subsequent hyper-parameters such as $g$ and $\tau^2$  which do not appear in the likelihood. 

A main reason for the popularity of the conjugate prior framework is that it often allows for marginalization over $\vbeta$ and $\sigma^2$, resulting in closed form expressions for Bayes factors and updates of posterior model probabilities. This allowed for analyses of the model selection consistency \citep{BB12} as well as more computationally efficient MCMC algorithms \citep{GM97}.  Despite these advantages, however, the conjugate prior framework is not innocuous for variance estimation, as we will show in this work. 

Alternatively to the conjugate prior framework, one might treat $\vbeta$ and $\sigma^2$ as independent \emph{a priori}. The formulation corresponding to \eqref{conj} for this independence prior framework is:
\begin{align}
&\vbeta \sim N(0, \mV), \label{ind}\\
&\pi(\sigma) \propto 1/\sigma.\notag
\end{align} 
Note that the prior characterization \eqref{ind} does not yield a normal inverse-gamma posterior distribution on $(\vbeta, \sigma^2)$ and as such is not conjugate.

In addition to the above prior frameworks, Bayesian methods for variable selection can be further categorized by the way they treat negligible predictors. Discrete component Bayesian methods for variable selection exclude negligible predictors from consideration, adaptively reducing the dimension of $\vbeta$. Examples of such discrete component methods include spike-and-slab priors where the ``spike'' distribution is a point-mass at zero \citep{MB88}. In contrast, continuous Bayesian methods for variable selection shrink, rather than exclude, negligible predictors and as such $\vbeta$ remains $p$-dimensional \citep{GM93, PS10, RG14}.

In this paper, we show that for continuous Bayesian variable selection methods, the conjugate prior framework can result in underestimation of the error variance when: (i) the regression coefficients $\vbeta$ are sparse; and (ii) $p$ is of the same order as, or larger than $n$.
Intuitively, conjugate priors implicitly add $p$ ``pseudo-observations'' to the posterior which can distort inference for the error variance when the true number of non-zero $\vbeta$ is much smaller than $p$.  This is not the case for discrete component methods which adaptively reduce the size of $\vbeta$. To avoid the underestimation problem in the continuous case, we recommend the use of independent priors {on $\vbeta$ and $\sigma^2$}.  Further, we extend the Spike-and-Slab Lasso of \citet{RG18} to the unknown variance case with {an independent prior formulation,} and highlight the performance gains over the known variance case via a simulation study. On the protein activity dataset of \citet{CP98}, we demonstrate the benefit of simultaneous variance estimation for both variable selection and prediction. The implementation of the Spike-and-Slab Lasso is publicly available in the R package \texttt{SSLASSO} \citep{SSLpackage}.

It is important to note the difference in the scope of this work with previous work on variance priors, including \citet{G04, BB12, L08}. Here, we are focused on the estimation of the error variance, $\sigma^2$. In contrast, the aforementioned works are concerned with the choice of priors for hyper-parameters which do not appear in the likelihood, i.e. the $g$ in the $g$-prior, and $\tau^2$ and $\lambda_j^2$ for global-local priors.  We recognize the importance of the choice of these priors for Bayesian variable selection; however, the focus of this paper is the prior choice for the error variance in conjunction with variable selection. 

We also note that our discussion considers only Gaussian related prior forms for the regression coefficients.  Despite this seemingly limited scope, we note that the majority of priors used in Bayesian variable selection can be cast as a scale-mixture of Gaussians \citep{PS10}, and that popular frequentist procedures such as the Lasso and variants thereof also fall under this framework.

The paper is structured as follows.  In Section 2, we discuss invariance arguments for conjugate priors and draw connections with Jeffreys priors. We then highlight situations where we ought to depart from Jeffreys priors; namely, in multivariate situations. In Section 3, we take Bayesian ridge regression as an example to highlight why conjugate priors can be a poor choice. In Section 4, we draw connections between Bayesian regression and concurrent developments with variance estimation in the penalized likelihood literature. In Section 5, we examine the mechanisms of the Gaussian global-local shrinkage framework and illustrate why they can be incompatible with the conjugate prior structure. In Section 6, we consider the Spike-and-Slab Lasso of \citet{RG18} and highlight how the conjugate prior yields poor estimates of the error variance. We then extend the procedure to include {the unknown variance case} using an independent prior structure and demonstrate via simulation studies how this leads to performance gains over not only the known variance case, but a variety of other variable selection procedures. In Section 7, we apply the Spike-and-Slab Lasso with unknown variance to the protein activity dataset of \citet{CP98}, highlighting the improved predictive performance afforded by simultaneous variance estimation. {We conclude with a discussion in Section 8}.

\section{Invariance Criteria}

A common argument used in favor of the conjugate prior for Bayesian linear regression is
that it is invariant to scale transformations of the response (Bayarri et al.,
2012). That is, the regression coefficients depend \emph{a priori} on $\sigma^2$ in a ``scale-free way'' through
\begin{align}
\pi(\vbeta|\sigma^2) = \frac{1}{\sigma^p} h(\vbeta/\sigma),
\end{align}
for some proper density function $h(x)$.  This means that the units of measurement used for the response do not affect the resultant estimates; for example, if $\mY$ is scaled by a factor of $c$, one would expect that the estimates for the regression coefficients, $\vbeta$, and error variance, $\sigma^2$, should also be scaled by $c$.

A more general principle of invariance was proposed by \citet{J61} in his seminal work, {\em The Theory of Probability}, a reference which is also sometimes given for the conjugate prior. In this section,  we examine the original invariance argument of \citet{J61} and highlight a caveat with this principle that the author himself noted; namely that it should be avoided in multivariate situations.  We then draw connections between this suboptimal multivariate behavior and the conjugate prior framework, ultimately arguing similarly to Jeffreys that we should treat the mean and variance parameters as independently \emph{a priori}.

\subsection{Jeffreys Priors}
For a parameter $\alpha$, the {Jeffreys} prior is
\begin{equation}
\pi(\alpha) \propto |I(\alpha)|^{1/2},
\end{equation}
where $I(\alpha)$ is the Fisher information matrix.  The main motivation given by \citet{J61} for these priors was that they are invariant for all nonsingular transformations of the parameters. This property appeals to intuition regarding objectivity; ideally, the prior information we decide to include should not {depend} upon the choice of the parameterization, which itself is arbitrary. 

Despite this intuitively appealing property, the following problem with this principle was spotted in the original work of \citet{J61} and later re-emphasized by \citet{R09} in their revisit of the work. Consider the model
\begin{align*}
Y_i \sim N(\mu, \sigma^2), \quad  i = 1,\dots, n.
\end{align*}
If we treat the parameters $\mu$ and $\sigma$ independently, the Jeffreys priors are  $\pi(\mu) \propto 1$ and $\pi(\sigma) \propto 1/\sigma$.  However, if the parameters are considered jointly, the Jeffreys prior is $\pi(\mu, \sigma) \propto 1/\sigma^2$. This discrepancy is exaggerated when we include more parameters. {In effect}, by considering the parameters jointly as opposed to independently, we are implicitly including additional ``pseudo-observations'' of $\sigma^2$ and consequently distorting our estimates of the error variance.  

This ``pseudo-observation'' interpretation can be seen explicitly in the conjugate form of the Jeffreys prior for a Gaussian likelihood. For example: suppose now we have an $n$-dimensional mean denoted by $\vmu = (\mu_1, \dots, \mu_n)$. That is,
\begin{align*}
Y_i \sim N(\mu_i, \sigma^2), \quad   i = 1,\dots, n.
\end{align*}
The joint Jeffreys prior $\pi(\vmu, \sigma) \propto 1/\sigma^{n + 1}$ is an improper inverse-gamma prior with shape parameter, $n/2$, and scale parameter zero.  As the prior is conjugate, the posterior distribution for the variance is also inverse-gamma:
\begin{equation}
\pi(\sigma^2|\mY, \vmu) \sim IG\left(\frac{n}{2} + \frac{n}{2}, \ 0 + \frac{\sum_{i=1}^n (Y_i - \mu_i)^2}{2}\right) \label{pseudo_observations}
\end{equation}
where the first term of both the shape and scale parameters in \eqref{pseudo_observations} are the prior hyperparameters. Thus, the dependent Jeffreys prior can be thought of as encoding knowledge of $\sigma^2$ from a previous experiment where there were $n$ observations which yielded a sample variance of zero.  This results in the prior concentrating around zero for large $n$ and will severely distort posterior estimates of $\sigma^2$.  As we shall see later, this \emph{dependent Jeffreys prior} for the parameters is in some cases akin to the conjugate prior framework in \eqref{conj}.

This prior dependence between the parameters is explicitly repudiated by \citet{J61} who states (with notation changed to match ours): ``in the usual situation in an estimation problem, $\mu$ and $\sigma^2$ are each capable of any value over a considerable range, and neither gives any appreciable information about the other. We should then take: $\pi(\mu, \sigma) = \pi(\mu)\pi(\sigma).$''  That is, Jeffreys' remedy is to treat the parameters independently \emph{a priori}, a recommendation which we also adopt.  In addition, Jeffreys points out that the key problem with the joint Jeffreys prior is that it does not have the same reduction of degrees of freedom required by the introduction of additional nuisance parameters. We shall examine this phenomenon in more detail in Section 3 where we will discuss the consequences of using dependent Jefferys priors and other conjugate formulations in Bayesian linear regression.

We note a possible exception to this independence argument which is found later in {\em The Theory of Probability} where Jeffreys argues that for simple normal testing, the prior on $\mu$ under the alternative hypothesis should depend on $\sigma^2$. However, it is important to note that this recommendation is for the situation where $\mu$ is one-dimensional and so the underestimation problem observed in \eqref{pseudo_observations} is not a problem.  Given Jeffreys' earlier concerns regarding multivariate situations, it is unlikely he intended this dependence to generalize for higher dimensional $\mu$.

\section{Bayesian Regression}

\subsection{Prior considerations}

Consider again the classical linear regression model in \eqref{regression}. For a non-informative prior, it is common to use $\pi(\vbeta, \sigma^2) \propto 1/\sigma^2$ \citep[see, for example,][]{G14}. Similarly to our earlier discussion, this prior choice corresponds to multiplying the independent, Jeffreys priors for $\vbeta$ and $\sigma$. In contrast, the joint Jeffreys prior would be $\pi(\vbeta, \sigma^2) \propto 1/\sigma^{p+2}$. Let us now examine the estimates resulting from the former, independent Jeffreys prior.  In this case, we have the following marginal posterior mean estimate for the error variance:
\begin{equation}
\E[\sigma^2|\mY] = \frac{\lVert \mY - \mX\widehat{\vbeta}\rVert^2}{n -p - 2}\label{ind_jeffreys_mean}
\end{equation}
where $\widehat{\vbeta} = (\mX^T\mX)^{-1}\mX^T\mY$ is the usual least squares estimator. We observe that the degrees of freedom adjustment, $n-p$, naturally appears in the denominator. This does not occur for the joint Jeffreys prior where the marginal posterior mean is given by:
\begin{equation}
\E[\sigma^2|\mY] =\frac{\lVert \mY - \mX\widehat{\vbeta}\rVert^2}{n-2}. \label{joint_jeffreys_mean}
\end{equation}
For large $p$, {this} estimator with the joint Jeffreys prior will severely underestimate the error variance.  Avoiding this, it is commonly accepted that the independent Jeffreys prior $\pi(\vbeta, \sigma^2) \propto 1/\sigma^2$ should be the default non-informative prior in this setting. 

There is no such clarity, however, in the use of conjugate priors for Bayesian linear regression. To add to this discourse, we show that these conjugate priors can suffer the same problem as the dependent Jeffreys priors and recommend, similarly to Jeffreys, that independent priors should be used instead.  We make this point with the following example. A common conjugate prior choice for Bayesian linear regression is 
\begin{equation}
\vbeta|\sigma^2, \tau^2 \sim N_p(0, \sigma^2\tau^2\mI). \label{ridge}
\end{equation}
If we consider the parameter $\tau^2$ to be fixed, this prior choice corresponds to Bayesian ridge regression. {With} an additional non-informative prior $\pi(\sigma^2)\propto 1/\sigma^2$, we then have the joint prior
\begin{equation}
\pi(\vbeta|\sigma^2)\pi(\sigma^2) = \pi(\vbeta, \sigma^2) \propto \frac{1}{\sigma^{p + 2}}\exp\left\{-\frac{1}{2\sigma^2\tau^2}\lVert \vbeta\rVert^2\right\}. \label{conj_ridge}
\end{equation}
Note again the $\sigma^{p + 2}$ in the denominator, similarly to the joint Jeffreys prior. 

It is illustrative to consider the conditional prior dependence here of $\sigma^2$ on $\vbeta$ from a ``pseudo-observation'' perspective: the implicit conditional prior on $\sigma^2$ from \eqref{conj_ridge} is given by
\begin{align}
\sigma^2 |\vbeta \sim IG\left(\frac{p }{2}, \frac{\lVert \vbeta\rVert^2}{2\tau^2}\right). \label{sigma_ig}
\end{align}
Similarly to the discussion in Section 2, this inverse-gamma prior has the following interpretation: from a previous experiment, the sample variance of $p$ observations was $\frac{1}{p}\lVert \vbeta \rVert^2/\tau^{2}$.  For regions where $\vbeta$ is sparse, this dependence leads to prior concentration of $\sigma^2$ around zero as illustrated by the following.
\begin{proposition}\label{conjugate_concentration}
Suppose $\lVert \vbeta \rVert_0 = q$ and $\max_j \beta_j^2 = K$ for some constant $K\in\R$ .  Denote the true variance as $\sigma_0^2$. Then
\begin{align}
P\left({\sigma^2}/{\sigma_0^2} \geq \varepsilon \ |\ \vbeta\right) &\leq \frac{q}{p-2}\frac{K}{\tau^2}\frac{1}{\varepsilon \sigma_0^2}.\label{markov_1}
\end{align}
 \end{proposition}
\begin{proof}
Proposition \ref{conjugate_concentration} follows from Markov's inequality and the bound $\lVert \vbeta\rVert^2 \leq qK$.
\end{proof}

Proposition \ref{conjugate_concentration} implies that as $q/p \to 0$, we can choose $0 < \varepsilon < 1$ such that the prior places decreasing mass on values of $\sigma^2$ greater than $\varepsilon\sigma_0^2$. Thus, in regions of bounded sparse regression coefficients, the conjugate Gaussian prior can result in poor estimation of the true variance.

From a more philosophical perspective, \eqref{sigma_ig} corresponds to prior knowledge that the error variance is implicitly the sample variance of previous observations of the regression coefficients, $\vbeta$.  This is troubling given that the error variance is independent of the signal and in particular of the regression coefficients. 

In the next section, we conduct a simulation study for the simple case of Bayesian ridge regression and show empirically how this implicit prior on $\sigma^2$ can distort estimates of the error variance. 

\subsection{The failure of a conjugate prior}

As an illustrative example, we take $n = 100$ and $p = 90$ and compare the least squares estimates of $\vbeta$ and $\sigma^2$ to Bayesian ridge regression estimates with (i) the conjugate formulation with \eqref{ridge} and (ii) the independent prior formulation with
\begin{equation}
\pi(\vbeta) \sim N_p(0, \tau^2 \mI).
\end{equation}
For both Bayesian ridge regression procedures we use the non-informative error variance prior: $\pi(\sigma^2) \propto 1/\sigma^2$. The predictors $\mX_i$, $i = 1,\dots, p$ are generated as independent standard normal random variables. The true $\vbeta_0$ is set to be sparse with only six non-zero elements; the non-zero coefficients are set to $\{-2.5, -2, -1.5, 1.5, 2, 2.5\}$.   The response $\mY$ is generated according to \eqref{regression} with the true variance being $\sigma^2 = 3$.  We take $\tau = 10$ {as known} and highlight that this weakly informative choice leads to poor variance estimates in the conjugate prior framework.  Whilst an empirical or fully Bayes approach for estimating $\tau^2$ may be preferable for high-dimensional regression, it is troubling that the conjugate prior yields poor results for a simple example where $n > p$ and in which least squares and the independent prior perform well. 

The conjugate prior formulation allows for the exact expressions for the marginal posterior means of $\vbeta$ and $\sigma^2$:
\begin{align}
\E[\vbeta|\mY] &= [\mX^T\mX + \tau^{-2}\mI]^{-1}\mX^T\mY  \label{conj_means} \\
  \E[\sigma^2|\mY] &= \frac{\mY^T[\mI-\mH_{\tau}]\mY}{n - 2} \label{conj_sigma_marginal}
\end{align}
where $\mH_{\tau} = \mX [\mX^T\mX + \tau^{-2}\mI]^{-1}\mX^T$. Similarly to \eqref{joint_jeffreys_mean}, the above marginal posterior mean for $\sigma^2$ does not incorporate a degrees of freedom adjustment and so we expect this estimator to underestimate the true error variance. 

It is illuminating to observe the underestimation problem when considering the conditional posterior mean of $\sigma^2$, instead of the marginal:
\begin{equation}
\E[\sigma^2|\mY, \vbeta] = \frac{\lVert \mY - \mX\vbeta\rVert^2 + \lVert \vbeta\rVert^2/\tau^2}{n + p - 2}.  \label{conj_joint_sigma}
\end{equation}
The additional $p$ in the denominator here leads to severe underestimation of $\sigma^2$ when $\vbeta$ is sparse and bounded as in Proposition \ref{conjugate_concentration} and $p$ is of the same order as, or larger than, $n$, as discussed in the previous section. We note in passing that a value of $\tau^2$ close to $\lVert \vbeta\rVert^2/p\sigma^2$, which may be obtainable with an empirical or fully Bayes approach, would avoid this variance underestimation problem, as can be seen from \eqref{conj_joint_sigma}.

This is in contrast to the conditional posterior mean for $\sigma^2$ using the independent prior formulation \eqref{ind}, which we also consider. This estimator is given by:
\begin{equation}
\E[\sigma^2|\mY, \vbeta] = \frac{\lVert \mY-\mX\vbeta\rVert^2}{n-2}.\label{ind_joint_sigma}
\end{equation}
Here we do not observe a degrees of freedom adjustment because \eqref{ind_joint_sigma} is the \emph{conditional} posterior mean, not the marginal. Earlier in \eqref{ind_jeffreys_mean} we considered the marginal posterior mean for the independent Jeffreys' prior which led to the $n-p$ in the denominator.  For the marginal posterior means of $\vbeta$ and $\sigma^2$, the independent prior formulation does not yield closed form expressions. To assess these, we use a Gibbs sampler, the details of which may be found in the appendix.  

When $\tau^2$ is large, the estimate of $\vbeta$ for both the conjugate and independent formulations are almost exactly the least-squares estimate, 
$\widehat{\vbeta} = [\mX^T\mX]^{-1}\mX^T\mY.$  However, the estimates of the variance $\sigma^2$ differ substantially.

In Figure \ref{ridge_boxplot}, we display a boxplot of the estimates of $\sigma^2$ for (i) Least Squares, (ii) Conjugate Bayesian ridge regression, (iii) Zellner's  prior:
\begin{equation}
\vbeta|\sigma^2 \sim N(0, \sigma^2\tau^2[\mX^T\mX]^{-1}),
\end{equation}
and (iv) Independent Bayesian ridge regression over 100 replications.  Here the estimates from least squares and the independent ridge are similarly centered around the truth; however, the conjugate ridge and Zellner's priors consistently underestimate the error variance with medians of $\widehat{\sigma}^2 = 0.27$ and $0.55$, respectively. This poor performance is a direct result of the bias induced by adding $p$ ``pseudo-observations'' of $\sigma^2$ as discussed in Section 3.1, which also occurs for the Zellner prior. 

\begin{figure}[ht]
\centering
\includegraphics[width = 0.6\textwidth]{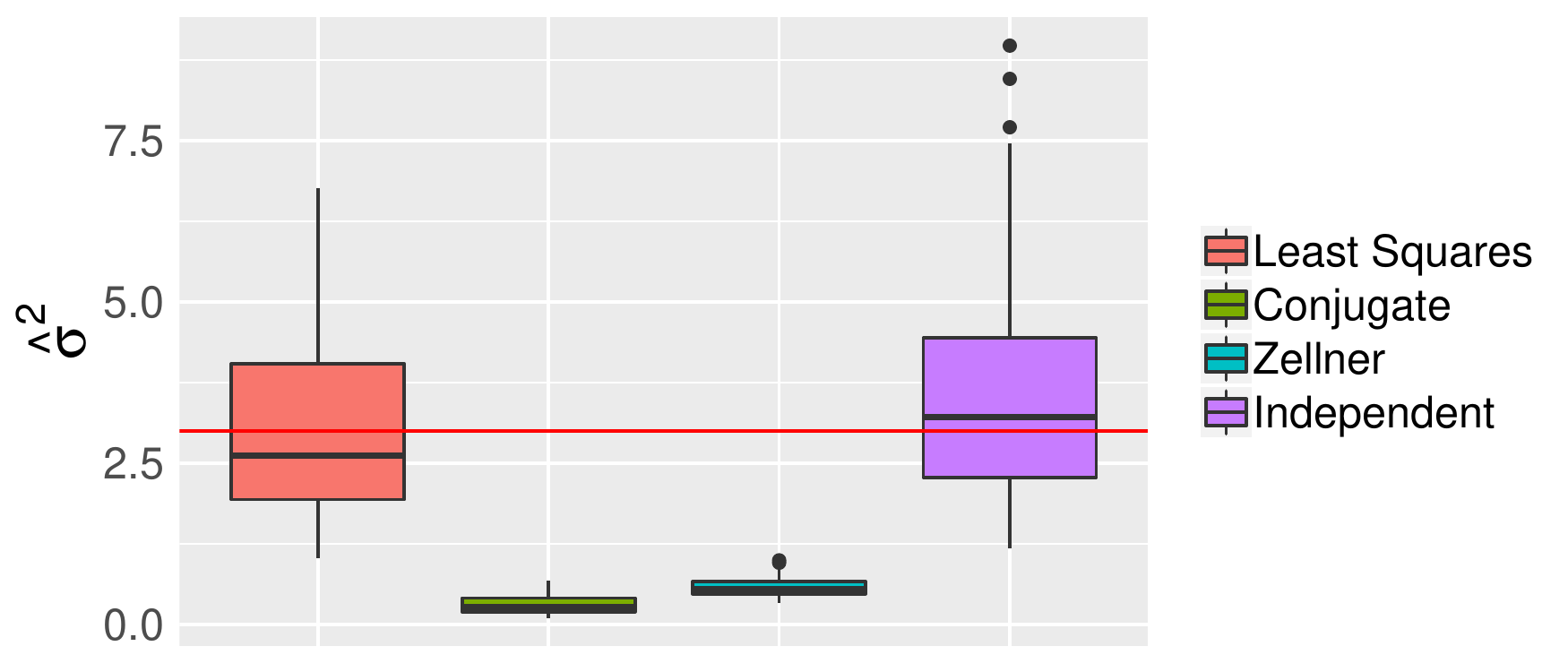}
\caption{Estimated $\widehat{\sigma}^2$ for each procedure over 100 repetitions. The true $\sigma^2 =3$ is the red horizontal line.} \label{ridge_boxplot}
\end{figure}

This phenomenon of underestimating $\sigma^2$ is also seen in EMVS \citep{RG14}, which can be viewed as iterative Bayesian ridge regression with an adaptive penalty term  for each regression coefficient $\beta_j$ instead of the same $\tau^2$ above. EMVS also uses a conjugate prior formulation in which $\vbeta$ depends on $\sigma^2$ \emph{a priori} similarly to \eqref{ridge}.  As in the above ridge regression example, with this prior EMVS yields good estimates for $\vbeta$, but severely underestimates $\sigma^2$.  This {is evident in the Section 4 example} of \citet{RG14} with $n= 100$ and $p=1000$.  There, conditionally on the modal estimate of $\beta$, the associated modal estimate of $\sigma^2$ is 0.0014, a severe underestimate of the true variance $\sigma^2 = 3$. Fortunately, EMVS can be easily modified to use the independent prior specification, as now has been done in the publicly available \texttt{EMVS} R package \citep{EMVSpackage}.  It is interesting to note that the SSVS procedure of \citet{GM93} used the nonconjugate independence prior formulation in lieu of the conjugate prior formulation for the continuous spike-and-slab setup.

A natural question to ask is: how does the poor estimate of the variance in the conjugate case affect the estimated regression coefficients? Insight is obtained by comparing \eqref{conj_means} to the conditional posterior mean of $\vbeta$ in the independent case, given by:
\begin{equation}
\E[\vbeta|\sigma^2, \mY] =\left[\mX^T\mX+ \frac{\sigma^2}{\tau^2}\mI\right]^{-1}\mX^T\mY.\label{ind_mean}
\end{equation}
In \eqref{conj_means}, the Gaussian prior structure allows for $\sigma^2$ to be factorized out so that the estimate of $\vbeta$ does not depend on the variance.  This lack of dependence on the variance is troubling, however, as we want to select fewer variables when the error variance is large making the {signal-to-noise} ratio low. This is in contrast to \eqref{ind_mean} where when $\sigma^2$ is large relative to $\tau^2$, the signal-to-noise ratio is low and so the posterior estimate for $\vbeta$ will be close to zero, reflecting the relative lack of information. This does not occur for the posterior mean of $\vbeta$ in the conjugate case. 

\subsection{What about a prior degrees of freedom adjustment?}

At this point, one may wonder: if the problem seems to be the extra $\sigma^p$ in the denominator, why not use the prior $\pi(\sigma^2) \propto \sigma^{p - 4}$ instead of the right-Haar prior $\pi(\sigma^2) \propto \sigma^{-2}$ that is commonly used? This ``$p$-sigma'' prior then results in the joint prior:
\begin{equation}
\pi(\vbeta|\sigma^2)\pi(\sigma^2) \propto \frac{1}{(\sigma^2)^2}\exp\left\{-\frac{1}{2\sigma^2\tau^2}\lVert \vbeta \rVert^2\right\}. 
\end{equation}
We can again consider the implicit conditional prior on $\sigma^2$:
\begin{align}
\sigma^2|\vbeta \sim IG\left(1, \frac{\lVert \vbeta\rVert^2}{2\tau^2}\right).
\end{align}

For the simulation setup in section 3.2, this alternative conjugate prior would in fact remedy the variance estimates of the conjugate formulation \eqref{ridge}.  However, the $p$-sigma prior suffers from other drawbacks.

In particular, the $p$-sigma prior can actually lead to \emph{overestimation} of the error variance, as opposed to the underestimation by the conjugate prior formulation \eqref{conj} observed in section 3.2.  This overestimation can be seen from the concentration of the prior captured in Proposition \ref{p_sigma_concentration} below.  A similar overestimation phenomenon was also found for a penalized likelihood procedure which implicitly uses a $p$-sigma prior, as we will dicuss in section 4.

\begin{proposition} \label{p_sigma_concentration}
Suppose $\lVert \vbeta \rVert_0 = q$ and $\min_{j, \beta_j\neq 0} \beta_j^2 = K$ for some constant $K\in\R$.  Denote the true variance as $\sigma_0^2$. Then
\begin{align}
P(\sigma^2/\sigma_0^2 \geq \varepsilon|\ \vbeta) \geq 1 - \exp\left(-\frac{qK}{2\varepsilon \sigma_0^2\tau^2}\right).
\end{align}
\end{proposition}
\begin{proof}
We have:
\begin{align*}
P(\sigma^2 \geq \varepsilon\sigma_0^2|\ \vbeta) &= \int_{\varepsilon\sigma_0^2}^{\infty} \frac{\lVert \vbeta\rVert^2}{2\tau^2} \frac{1}{u^2}\exp\left(-\frac{\lVert \vbeta\rVert^2}{2\tau^2}\frac{1}{u}\right) du \\
&=1- \exp\left(-\frac{\lVert \vbeta\rVert^2}{2\tau^2\varepsilon \sigma_0 2}\right) \\
&\geq 1 - \exp\left(-\frac{qK}{2\tau^2\varepsilon\sigma_0^2}\right).
\end{align*}\end{proof}

Proposition \ref{p_sigma_concentration} implies that as $q\to \infty$, we can choose arbitrary $\varepsilon > 0$ such that $\sigma^2$ will overestimate the true variance.  As many posterior concentration results require $q \to \infty$, albeit at a much slower rate than $p$ \citep[see, for example,][]{VDP16}, this is particularly troublesome.

Another concern regarding the $p$-sigma prior is more philosophical. As $p$ gets larger, the $p$-sigma prior puts increasing mass on larger and larger values of $\sigma^2$, which does not seem justifiable.

In contrast, the only drawback of the independent prior form is that it can be more computationally intensive. This was seen in Section 3.1 where the independent form did not yield closed form expressions for the posterior means. However, most variable selection and shrinkage problems today use more complicated global-local prior forms for the regression coefficients which are also computationally intensive, no matter whether one uses a conjugate or independent prior.  

For these reasons, we prefer the independent prior forms for the regression coefficients and error variance. We are also of the opinion that the simplicity of the independent prior is in its favor.

\section{Connections with Penalized Likelihood Methods}
Here we pause briefly to examine connections between Bayesian methods and developments in estimating the error variance in the penalized regression literature.  Such connections can be drawn as penalized likelihood methods are implicitly Bayesian; the penalty functions can be interpreted as priors on the regression coefficients so these procedures also in effect yield MAP estimates. 

One of the first papers to consider the unknown error variance case for the Lasso was \citet{S10}, who suggested the following penalized loss function for introducing unknown variance into the frequentist Lasso framework:
\begin{equation}
L_{pen}(\vbeta, \sigma^2) = \frac{\lVert \mY -\mX\vbeta\rVert^2}{2\sigma^2} + \frac{\lambda}{\sigma}\lVert \vbeta \rVert_1 + n\log \sigma. \label{stadler}
\end{equation}
Optimizing this objective function is in fact equivalent to MAP estimation for the following Bayesian model with the $p$-sigma prior discussed in Section 3.2:
\begin{align}
\mY &\sim N(\mX\vbeta, \sigma^2\mI) \\
\pi(\vbeta) &\propto \frac{1}{\sigma^p}\prod_{j=1}^p e^{-\lambda|\beta_j|/\sigma} \notag \\
\pi(\sigma^2) &\propto \sigma^p.\notag
\end{align}
Interestingly, \citet{SZ10} proved that the resulting estimator for the error variance \emph{overestimates} the noise level unless $\lambda\lVert\vbeta^*\rVert_1 /\sigma^* = o(1)$, where $\vbeta^*$ and $\sigma^*$ are the true values of the regression coefficients and error variance, respectively.   Let us examine this condition more closely. Suppose again the true dimension of $\vbeta^*$ is $q$, where $q \ll p$, and $\lambda = A\sqrt{2/n\log p}$ for some constant $A > 1$ (the $\lambda$ required by \citet{SZ12} for consistency). Suppose also that $\max_j |\beta_j| = K$ for some constant $K\in \R$. Then, the error variance estimate from this procedure will be upwardly biased unless 
\begin{equation*}
\lambda\lVert \vbeta^*\rVert \approx qK\sqrt{2/n \log p}  = o(1).
\end{equation*}
That is, the true dimension $q$ cannot at the same time increase at the required rate for posterior contraction \emph{and} result in consistent estimates for the error variance.  Note also the connection to Proposition \ref{p_sigma_concentration} - there, the prior mass on $\sigma^2$ will concentrate on values greater than the true variance $\sigma_0^2$ unless
$\lVert \vbeta\rVert^2/\tau^2 = o(1)$.

To resolve this issue of overestimating the error variance, \citet{SZ12} proposed as an alternative the ``scaled Lasso'', an algorithm which minimizes the following penalized joint loss function via coordinate descent:
\begin{equation}
L_{\lambda}(\vbeta, \sigma) = \frac{\lVert\mY - \mX\vbeta\rVert^2}{2\sigma} + \frac{n\sigma}{2} + \lambda\sum_{j=1}^p |\beta_j|. \label{scaled_lasso}
\end{equation}
This loss function is a penalized version of Huber's concomitant loss function, and so may be viewed as performing robust high-dimensional regression. It is also equivalent to the ``square-root Lasso'' of \citet{B14}.  Minimization of the loss function \eqref{scaled_lasso} can be viewed as MAP estimation for the Bayesian model (with a slight modification):
\begin{align}
\mY &\sim N(\mX\vbeta, \sigma\mI) \label{scaled_lasso_bayesian} \\
\pi(\vbeta) &\propto \prod_{j=1}^p \frac{\lambda}{2}e^{-\lambda|\beta_j|} \notag\\
\sigma &\sim \text{Gamma}(n+1, n/2). \notag
\end{align}
Note that to interpret the scaled Lasso as a Bayesian procedure, $\sigma$, rather than $\sigma^2$, plays the role of the variance in  \eqref{scaled_lasso_bayesian}.  \citet{SZ12} essentially then re-interpret $\sigma$ as the standard deviation again after optimization of \eqref{scaled_lasso}. This re-interpretation can be thought of as an ``unbiasing'' step for the error variance.  It is a little worrisome, however, that the implicit prior on the error variance is very informative: as $n\to\infty$, this Gamma prior concentrates around $\sigma = 2$. 

\citet{SZ12} proved that the scaled Lasso estimate $\widehat{\sigma}(\mX, \mY)$ is consistent for the ``oracle'' estimator 
\begin{equation}
\sigma^* = \frac{\lVert \mY - \mX\vbeta^*\rVert}{\sqrt{n}}, \label{oracle_sigma}
\end{equation}
where $\vbeta^*$ are the true regression coefficients, for the value of $\lambda_0 \propto \sqrt{2/n\log p}$.  This estimator \eqref{oracle_sigma} is called the oracle because it treats the true regression coefficients as if they were known. The term $\lVert \mY - \mX\vbeta^*\rVert^2$ is then simply the sum of normal random variables, of which we calculate the variance as $\sum_{i=1}^n\varepsilon_i^2/n$. 

More recently, \citet{SZ13} proposed a different value of $\lambda_0$ which achieves tighter error bounds than in their previous work. Specifically, they propose 
\begin{equation}
\lambda_0 =\sqrt{2}L_n(k/p)
\end{equation}
with $L_n(t) = \frac{1}{\sqrt{n}}\Phi^{-1}(1-t)$ where $\Phi$ is the standard Gaussian cdf and $k$ is the solution to $k = L_1^4(k/p) + 2L_1^2(k/p)$.

\section{Global-Local Shrinkage}
In this section, we examine how {the use of a conjugate prior} affects the machinery of the Gaussian global-local shrinkage paradigm.  The general structure for this class of priors is given by:
\begin{align}
\beta_j &\sim N(0, \tau^2\lambda_j^2), \quad \lambda_j^2 \sim \pi(\lambda_j^2),\quad j = 1,\dots, p \label{gl_setting}\\
\tau^2&\sim \pi(\tau^2)\notag
\end{align}
where $\tau^2$ is the ``global'' variance and $\lambda_j^2$ is the ``local'' variance. 
Note that taking $\tau^2$ to be the same as the error variance $\sigma^2$ would result in a conjugate prior in this setting.  This is exactly what was done in the original formulation of Bayesian lasso by \citet{PC08}, which can be recast in the Gaussian global-local shrinkage framework as follows (notation changed slightly for consistency):
\begin{align}
\mY |\vbeta, \sigma^2 &\sim N_n(\mX\vbeta, \sigma^2 \mI_n) \label{bayesian_lasso}\\
\beta_j|\sigma^2, \lambda_j^2 &\sim N(0, \sigma^2\lambda_j^2),\quad \pi(\lambda_j^2) =\frac{u^2}{2}e^{-u^2\lambda_j^2/2},\quad  j = 1,\dots, p  \notag\\
\pi(\sigma^2) &\propto \sigma^{-2}. \notag
\end{align}
In the conjugate formulation in \eqref{bayesian_lasso}, $\sigma^2$ plays the dual role of representing the error variance as well as acting as the global shrinkage parameter.  This is problematic in light of the mechanics of global-local shrinkage priors.  Specifically, \citet{PS10} give the following requirements for the global and local variances in \eqref{gl_setting}: $\pi(\tau^2)$ should have substantial mass near zero to shrink all the regression coefficients so that the vast majority are negligible; and $\pi(\lambda_j^2)$ should have heavy tails so that it can be quite large, allowing for a few large coefficients to ``escape'' the heavy shrinkage of the global variance.
 
This heuristic is formalized in much of the shrinkage estimation theory. For the normal means problem where $\mX = \mI_n$ and $\vbeta\in\R^n$,  \citet{VDP16} prove that the following condition results in the posterior recovering nonzero means with the optimal rate:
\begin{enumerate}[(i)]
\item $\pi(\lambda_j^2)$ should be a uniformly regular varying function which does not depend on $n$; and
\item $\tau^2 = \frac{q}{n}\log(n/q)$, where $q$ is number of non-zero $\beta_j$.
\end{enumerate}
The uniformly regular varying property in (i) intuitively preserves the ``flatness'' of the prior even under transformations of the parameters, unlike traditional ``non-informative'' priors \citep{B16}. In preserving these heavy tails, such priors for $\lambda_j^2$ allow for a few large coefficients to be estimated.   The condition (ii) encourages $\tau^2$ to tend to zero which would be a concerning property if it were also the error variance.  These results suggest we cannot identify the error variance with the global variance parameter on the regression coefficients as in \eqref{bayesian_lasso}: it cannot simultaneously both shrink all the regression coefficients and be a good estimate of the residual variance. Finally, we note that \citet{H09} also considered the independent case for the Bayesian lasso in which the error variance is not identified with the global variance. 

An alternative conjugate formulation for Gaussian global-local shrinkage priors is to instead include three variance terms in the prior for $\beta_j$: the error variance, $\sigma^2$, the global variance, $\tau^2$, and the local variance, $\lambda_j^2$.  For example, \citep{C10} give the conjugate form of the horseshoe prior:
\begin{align}
\beta_j|\sigma^2,\tau^2, \lambda_j^2 &\sim N(0, \sigma^2\tau^2\lambda_j^2), \quad \lambda_j^2 \sim \pi(\lambda_j^2),\quad j = 1,\dots, p  \label{conj_horseshoe}\\
\tau^2&\sim \pi(\tau^2), \notag\\
 \pi(\sigma^2) &\propto \sigma^{-2}.\notag
\end{align}
This prior formulation \eqref{conj_horseshoe} remedies the aforementioned issue in the Bayesian lasso as it separates the roles of the error variance and global variance. However, this prior structure can still be problematic for error variance estimation. 

Consider the conditional posterior mean of $\sigma^2$ for the model \eqref{conj_horseshoe}:
\begin{align}
\E[\sigma^2|\mY, \vbeta, \tau^2, \lambda_j^2] = \frac{\lVert \mY - \mX\vbeta\rVert^2 + \sum_{j=1}^p\beta_j^2/\lambda_j^2\tau^2}{n + p - 2}.\label{sigma_horseshoe}
\end{align}
Proposition \ref{global_local_sigma} highlights that, given the true regression coefficients, the conditional posterior mean of $\sigma^2$ underestimates the oracle variance \eqref{oracle_sigma} when $\vbeta$ is sparse.

\begin{proposition}\label{global_local_sigma}
Consider the global-local prior formulation given in \eqref{conj_horseshoe}. Denote the true vector of regression coefficients by $\vbeta^*$ where $\lVert \vbeta^*\rVert_0 = q$.  Suppose $\max_j\beta_j^{*2} = M_1$ for some constant $M_1\in\R$.  Denote the oracle estimator for $\sigma$ given in \eqref{oracle_sigma} by $\sigma^{*}$ and suppose $\sigma^* = O(1)$. Suppose also that for $j\in \{1, \dots, p\}$ with $\beta_j\neq0$, we have $\tau^2\lambda_j^2 > M_2$ for some $M_2 \in \R$.
Then
\begin{align}
\E[\sigma^2|\mY, \vbeta^*, \tau^2, \lambda_j^2] \leq \frac{n\sigma^{*2}}{n+p-2} + \frac{q M_1/M_2}{n + p - 2}. \label{global_local_lower_bound}
\end{align}
In particular, as $p/n \to \infty$ and $q/p\to 0$, we have
\begin{align}
\E[\sigma^2|\mY, \vbeta^*, \tau^2, \lambda_j^2] = o(1). 
\end{align} 
\end{proposition}
Given the mechanics of global-local shrinkage priors, the assumption  in Proposition \ref{global_local_sigma} that the the term $\tau^2\lambda_j^2$ is bounded from below for non-zero $\beta_j$  is not unreasonable. This is because for large $\beta_j$, the local variance $\lambda_j^2$ must be large enough to counter the extreme shrinkage effect of $\tau^2$. Indeed, the prior for $\lambda_j^2$ must have ``heavy enough'' tails to enable this phenomenon.

We should note that Proposition \ref{global_local_sigma} illustrates the poor performance of the posterior mean \eqref{sigma_horseshoe} given the true regression coefficients $\vbeta^*$, whereas the horseshoe procedure does not actually threshold the negligible $\beta_j$ to zero in the posterior mean of $\vbeta$.  For these small $\beta_j$, the term $\tau^2\lambda_j^2$ may be very small and potentially counteract the underestimation phenomenon. However, it is still troubling to use an estimator for the error variance that does not behave as the oracle estimator when the true regression coefficients are known. This is in contrast to the independent prior formulation where the conditional posterior mean of $\sigma^2$ is simply:
\begin{equation}
\E[\sigma^2|\mY, \vbeta] = \frac{\lVert \mY - \mX\vbeta\rVert^2}{n - 2}. 
\end{equation}

Note also that the problem of underestimation of $\sigma^2$ is exacerbated for modal estimation under the prior \eqref{conj_horseshoe}. This is because modal estimators often threshold small coefficients to zero and so the term $\sum_{j=1}^p \beta_j^2/\lambda_j^2\tau^2$  becomes negligible as in Proposition \ref{global_local_sigma}.  As MAP estimation using global-local shrinkage priors is becoming more common \citep[see, for example,][]{B17}, we caution against the use of these conjugate prior forms.

A different argument for using conjugate priors with the horseshoe is given by \citet{PV17}. They advocate for the model \eqref{conj_horseshoe}, arguing that it leads to a prior on the effective number of non-zero coefficients which does not depend on $\sigma^2$ and $n$. However, this quantity is derived from the posterior of $\vbeta$ and so does not take into account the uncertainty inherent in the variable selection process.  As a thought experiment: suppose that we know the error variance, $\sigma^2$, and number of observations, $n$. If the error variance is too large and the number of observations are too few, we would not expect to be able to say much about $\vbeta$ at all, and this intuition should be reflected in the effective number of non-zero coefficients. This point is similar to our discussion at the end of Section 3.2 regarding estimation of $\vbeta$. 

As before, we recommend {independent priors} on both the error variance and regression coefficients to both prevent distortion of the global-local shrinkage mechanism and to obtain better estimates of the error variance.

\section{Spike-and-Slab Lasso with Unknown Variance}
\subsection{Spike-and-Slab Lasso}
We now turn to the Spike-and-Slab Lasso \citep[SSL,][]{RG18} and consider how to incorporate the unknown variance case.  As the name suggests, the SSL involves placing a mixture prior on the regression coefficients $\vbeta$, where each $\beta_j$ is assumed \emph{a priori} to be drawn from either a Laplacian ``spike'' concentrated around zero (and hence be considered negligible), or a diffuse Laplacian ``slab" (and hence may be large).
Thus the hierarchical prior over $\vbeta$ and the latent indicator variables $\vgamma = (\gamma_1,\dots, \gamma_p)$ is given by
\begin{align}
\pi(\vbeta|\vgamma) &\sim \prod_{j=1}^p \left[\gamma_j\varphi_1(\beta_j) + (1-\gamma_j)\varphi_0(\beta_j)\right] \label{prior}, \\
\pi(\vgamma|\theta) &= \prod_{j=1}^p \theta^{\gamma_j}(1-\theta)^{1-\gamma_j} \quad\text{and} \quad \theta \sim \text{Beta}(a, b),
\end{align}
where $\varphi_1(\beta) = \frac{\lambda_1}{2}e^{-|\beta|\lambda_1}$ is the slab distribution and $\varphi_0(\beta) = \frac{\lambda_0}{2}e^{-|\beta|\lambda_0}$ is the spike ($\lambda_1 \ll \lambda_0$), and we have used the common exchangeable beta-binomial prior for the latent indicators.

We note that despite the use of the spike-and-slab prior typically associated with ``two-group'' Bayesian variable selection methods, the Spike-and-Slab Lasso can also be seen as a ``one-group'' method as the spike density is continuous.   Similarly to the Bayesian lasso, the Spike-and-Slab Lasso can be cast in the Gaussian global-local shrinkage framework using the Gaussian scale-mixture representation of the Laplace density:
\begin{align}
\pi(\vbeta|\tau_j^2) &\sim N(0, \tau_j^2)\notag \\
\pi(\tau_j^2|\gamma_j) &= \gamma_j\frac{\lambda_1^2}{2}e^{-\lambda_1\tau_j^2/2} + (1-\gamma_j)\frac{\lambda_0^2}{2}e^{-\lambda_0\tau^2_j/2}. \label{ssl_global_local}
\end{align}
Although in formulation \eqref{ssl_global_local} there is only a local variance and not a global variance, \citet{VDP16} show that the Spike-and-Slab Lasso can be interpreted similarly to global-local shrinkage priors. In this interpretation, the parameter $\theta$ (the proportion of non-zero $\beta_j$) essentially plays the role of the global-variance in that $\theta$ shrinks the majority of $\vbeta$ to zero.

\citet{RG18} recast this hierarchical model into a penalized likelihood framework, allowing for the use of existing efficient algorithms for modal estimation while retaining the adaptivity inherent in the Bayesian formulation. The regression coefficients $\vbeta$ are then estimated by
\begin{equation}
 \widehat{\vbeta} = \argmax_{\vbeta \in \R^p}\left\{-\frac{1}{2}\lVert \mY - \mX\vbeta\rVert^2 + pen(\vbeta) \right\} \label{ssl_fixed}
 \end{equation}
where 
\begin{equation}
pen(\vbeta) = \log\left[\frac{\pi(\vbeta)}{\pi(\vzero_p)}\right], \quad \pi(\vbeta) = \int_0^1 \prod_{j=1}^p [\theta \psi_1(\beta_j) + (1-\theta)\psi_0(\beta_j) ]d\pi(\theta).
\end{equation}

\citet{RG18} note a number of advantages in using a mixture of Laplace densities in \eqref{prior}, instead of the usual mixture of Gaussians as has been standard in the Bayesian variable selection literature. First, the Laplacian spike serves to automatically threshold modal estimates of $\beta_j$ to zero when $\beta_j$ is small, much like the Lasso. However, unlike the Lasso, the slab distribution in the prior serves to stabilize the larger coefficients so they are not downward biased. Additionally, the heavier Laplacian tails of the slab distribution yields optimal posterior concentration rates \citep{R18}.

A possible route for adding the unknown variance case to the SSL procedure is to follow the prior framework of \citet{PC08} in their Bayesian Lasso. There, \citet{PC08} used the following prior for the regression coefficients:
\begin{align}
\pi(\vbeta|\sigma^2) \propto \prod_{j=1}^p\frac{\lambda}{2\sigma}e^{-\lambda|\beta_j|/\sigma}.
\end{align}
In the next section, we illustrate why an analogous conjugate prior formulation for the Spike-and-Slab Lasso would be a poor choice. Afterwards, we introduce the SSL with unknown variance which utilizes an independent prior framework.

\subsection{The failure of the conjugate prior}

The conjugate prior formulation for the Spike-and-Slab Lasso is given by:
\begin{align}
\pi(\vbeta|\vgamma, \sigma^2) &\sim \prod_{j=1}^p \left(\gamma_j\frac{\lambda_1}{2\sigma}e^{-|\beta_j|\lambda_1/\sigma} + (1-\gamma_j)\frac{\lambda_0}{2\sigma}e^{-|\beta_j|\lambda_0/\sigma}\right)\\
\vgamma|\theta &\sim \prod_{j=1}^p \theta^{\gamma_j}(1-\theta)^{1-\gamma_j},\quad \theta\sim \text{Beta}(a,b)\\
p(\sigma^2) &\propto \sigma^{-2}. 
\end{align}
We find the posterior modes of our parameters using the EM algorithm, the details of which can be found in the appendix.  At the $(k+1)$th iteration, the EM updates are:
\begin{align}
\vbeta^{(k+1)} &=\argmin_{\vbeta}\left\{\frac{1}{2\sigma^{(k)}}\lVert \mY - \mX\vbeta\rVert^2   +\sum_{j=1}^p|\beta_j|\lambda^*(\beta_j^{(k)}/\sigma^{(k)}; \theta^{(k)})\right\} \label{ssl_conj}\\ 
\theta^{(k+1)}&= \frac{\sum_{j=1}^p p^*(\beta_j^{(k)}/\sigma^{(k)};\theta^{(k)}) + a - 1}{a+b+p-2}  \label{sigma_ssl_conj}\\
\sigma^{(k+1)} &= \frac{Q+ \sqrt{Q^2 +4(\lVert \mY- \mX\vbeta^{(k)}\rVert^2 )(n+p+2)} }
{2(n+p+2)}
\end{align}
where 
\begin{align}
Q &= \sum_{i=1}^p|\beta_j^{(k)}|\lambda^*(\beta_j^{(k)}/\sigma^{(k)}; \theta^{(k)}), \\
p^*(\beta; \theta) &= \left[1 + \frac{\lambda_0}{\lambda_1}\left(\frac{1-\theta}{\theta}\right)\exp\{-|\beta|(\lambda_0 - \lambda_1)\}\right]^{-1}, \label{pstar} \\
 \lambda^*(\beta; \theta) &= \lambda_1p^*(\beta; \theta) + \lambda_0(1-p^*(\beta;\theta)).  
\end{align}
Let us take a closer look at the estimator of $\sigma$.  Following the line of reasoning in \citet{SZ10}, an expert with oracle knowledge of the true regression coefficients $\vbeta^*$ would estimate the noise level by the oracle estimator:
\begin{equation}
\sigma^{*2} = \frac{\lVert \mY - \mX\vbeta^*\rVert}{n}.
\end{equation}
However, the maximum \emph{a posteriori} estimate of $\sigma$ at the true values of $\vbeta^*, \vgamma^*$ is given by
\begin{align}
\widehat{\sigma}_{MAP} = \tau + \sqrt{\tau^2 + \frac{(\sigma^*)^2}{1+p/n + 2/n} }
\end{align}
where $\tau = \lambda_1\lVert \vbeta^* \rVert_1/[2(n+p  + 2)]$. Here we see that if $n\to\infty$ with $p$ fixed, we have $\widehat{\sigma}_{MAP} \to \sigma^*$. If, however, we have $p/n \to \infty$ and $q/p \to 0$,  where the underlying sparsity  is $q=\lVert \vbeta^* \rVert_0$, we have $\widehat{\sigma}_{MAP} \to 0$.  Thus, similarly to the discussion of global-local shrinkage priors in Section 5, we will severely underestimate the error variance.  As before, the remedy is to use the independent prior on $\sigma^2$ and $\vbeta$. 

\subsection{Spike-and-Slab Lasso with Unknown Variance}

We now introduce the Spike-and-Slab Lasso with unknown variance, which considers the regression coefficients and error variance to be \emph{a priori} independent. The hierarchical model is
\begin{align}
\pi(\vbeta|\vgamma) &\sim \prod_{j=1}^p \left(\gamma_j\frac{\lambda_1}{2}e^{-|\beta_j|\lambda_1} + (1-\gamma_j)\frac{\lambda_0}{2}e^{-|\beta_j|\lambda_0}\right)\\
\vgamma|\theta &\sim \prod_{j=1}^p \theta^{\gamma_j}(1-\theta)^{1-\gamma_j},\quad \theta\sim \text{Beta}(a, b)\\
p(\sigma^2) &\sim \sigma^{-2}. 
\end{align}
The log posterior, up to an additive constant, is given by
\begin{equation}
L(\vbeta, \sigma^2) = -\frac{1}{2\sigma^2}\lVert \mY - \mX\vbeta\rVert^2 - (n+2)\log\sigma + \sum_{j=1}^p pen(\beta_j|\theta_j) \label{ssl_obj}
\end{equation}
where, for $j = 1,\dots, p$,
\begin{align}
pen(\beta_j|\theta_j) &= -\lambda_1|\beta_j| + \log[p^*(0;\theta_j)/p^*(\beta_j;\theta_j)], \label{ssl_pen}\\
\theta_j&= \E[\theta|\vbeta_{\backslash j}]
\end{align}
and $p^*(\beta;\theta)$ is as defined in \eqref{pstar}. For large $p$, \citet{RG18} note that the conditional expectation $\E[\theta|\vbeta_{\backslash j}]$ is very similar to $\E[\theta |\vbeta]$ and so for practical purposes we treat them as equal and denote ${\theta}_{\beta} = \E[\theta|\vbeta]$.

To find the modes of \eqref{ssl_obj}, we pursue a similar coordinate ascent strategy to \citet{RG18}, cycling through updates for each $\beta_j$ and $\sigma^2$ while updating the conditional expectation ${\theta}_{\beta}$.  This conditional expectation does not have an analytical expression; however, \citet{RG18} note that it can be approximated by
\begin{equation}
{\theta}_{\beta} \approx  \frac{a + \lVert \vbeta\rVert_0 }{a + b +p}.  \label{theta_approx}
\end{equation} 

We now outline the estimation strategy for $\vbeta$.  As noted in Lemma 3.1 of \citet{RG18}, there is a simple expression for the derivative of the SSL penalty:
\begin{equation}
\frac{\partial pen(\beta_j|\theta_{\beta})}{\partial |\beta_j|} \equiv - \lambda^*(\beta_j; {\theta_{\beta}})
\end{equation}
where
\begin{equation}
\lambda^*(\beta_j;{\theta}_{\beta}) = \lambda_1p^*(\beta_j;{\theta}_{\beta}) + \lambda_0[1-p^*(\beta_j;{\theta}_{\beta})].
\end{equation}
Using the above expression, the Karush-Kuhn-Tucker (KKT) conditions yield the following necessary condition for the global mode $\widehat{\vbeta}$:
\begin{equation}
\widehat{\beta}_j = \frac{1}{n}\left[|z_j| - \sigma^2\lambda^*(\widehat{\beta}_j;{\theta}_{\beta})\right]_+ sign(z_j), \quad j = 1,\dots, p\label{KKT}
\end{equation}
where $z_j = \mX_j^T(\mY - \sum_{k\neq j}^p \widehat{\beta}_k \cdot \mX_k)$ and we assume that the design matrix $\mX$ has been centered and standardized to have norm $\sqrt{n}$.  The condition \eqref{KKT} is very close to the familiar soft-thresholding operator for the Lasso, except that the penalty term $\lambda^*(\beta_j; \theta)$ differs for each coordinate.  Similarly to other non-convex methods, this enables \emph{selective shrinkage} of the coefficients, mitigating the bias issues associated with the Lasso. Also similarly to non-convex methods however, \eqref{KKT} is not a sufficient condition for the global mode. This is particularly problematic when the posterior landscape is highly multimodal, a consequence of  $p \gg n$ and large $\lambda_0$. To eliminate many of these suboptimal local modes from consideration, \citet{RG18} develop a more refined characterization of the global mode. This characterization follows the arguments of \citet{ZZ12} and can easily be modified for the unknown variance case of the SSL, detailed in Proposition \ref{global_mode}.
\begin{proposition}
The global mode $\widehat{\vbeta}$ satisfies
\begin{align}
\widehat{\beta}_j =
\begin{cases}
0 &\text{when } |z_j| \leq \Delta\\
\frac{1}{n}[|z_j| - \sigma^2\lambda^*(\widehat{\beta}_j;{\theta}_{\beta})]_+\text{sign}(z_j) &\text{when } |z_j| > \Delta
\end{cases}
\end{align}
where 
\begin{equation}
\Delta \equiv \inf_{t>0}[nt/2 - \sigma^2pen(t|\theta_{\beta})/t].\label{delta}
\end{equation}
\label{global_mode}
\end{proposition}

Unfortunately, computing \eqref{delta} can be difficult.  Instead, we seek an approximation to the threshold $\Delta$. A useful upper bound is $\Delta\leq \sigma^2\lambda^*(0;\theta_{\beta})$ \citep{ZZ12}. However, when $\lambda_0$ gets large, this bound is too loose and can be improved.  The improved bounds are given in Proposition \ref{delta_bounds}, the analogue of Proposition 3.2 of \citet{RG18} for the unknown variance case.   Before stating the result, the following function is useful to simplify exposition:
\begin{equation}
g(x;\theta) = [\lambda^*(x;\theta) - \lambda_1]^2 + \frac{2n}{\sigma^2}\log[p^*(x;\theta)].
\end{equation}
\begin{proposition}
When $\sigma(\lambda_0 - \lambda_1) > 2\sqrt{n}$ and $g(0;\theta_{\beta})>0$ the threshold $\Delta$ is bounded by
	$$\Delta^L < \Delta < \Delta^U,$$
	where
	\begin{align}
	\Delta^L &= \sqrt{2n\sigma^2\log[1/p^*(0;\theta_{\beta})] - \sigma^4 d_j} + \sigma^2\lambda_1, \\
	\Delta^U &= \sqrt{2n\sigma^2\log[1/p^*(0;\theta_{\beta})]} + \sigma^2\lambda_1
	\end{align}
	and $$0 < d_j < \frac{2n}{\sigma^2} - \left(\frac{n}{\sigma^2(\lambda_0 - \lambda_1)} - \frac{\sqrt{2n}}{\sigma}\right)^2.$$ 
	\label{delta_bounds}
\end{proposition}
Thus, when $\lambda_0$ is large and consequently $d_j \to 0$, the lower bound on the threshold approaches the upper bound, yielding the approximation $\Delta \approx \Delta^U$.  We additionally note the central role that the error variance plays in the thresholds in Proposition \ref{delta_bounds}. As $\sigma^2$ increases, the thresholds also increase, making it more difficult for regression coefficients to be selected. This is exactly what we want when the signal to noise ratio is small.

Bringing this all together, we incorporate this refined characterization of the global mode into the update for the coefficients via the generalized thresholding operator of \citet{M11}:
\begin{equation}
\widetilde{S}(z, \lambda, \Delta) = \frac{1}{n}(|z| - \lambda)_+sign(z)\mathbb{I}(|z| > \Delta).
\end{equation}
The coordinate-wise update is then 
\begin{equation}
\widehat{\beta}_j\leftarrow \widetilde{S}(z_j, \widehat{\sigma}^2\lambda^*(\widehat{\beta}_j;\widehat{\theta}_{\beta}), \Delta) \label{beta_update}
\end{equation}
where
\begin{equation}
\Delta =
\begin{cases}
 \sqrt{2n\widehat{\sigma}^2\log[1/p^*(0;\widehat{\theta}_{\beta})]} + \widehat{\sigma}^2\lambda_1 &\text{if } g(0;\widehat{\theta}_{\beta}) >0, \\
\widehat{\sigma}^2\lambda^*(0;\widehat{\theta}_{\beta}) &\text{otherwise.}
\end{cases} \label{delta_update}
\end{equation}
The conditional expectation ${\theta}_{\beta}$ is updated according to \eqref{theta_approx}.

Finally, given the most recent update of the coefficient vector $\widehat{\vbeta}$, the update for the error variance $\sigma^2$ is a simple Newton step:
\begin{align}
\widehat{\sigma}^2 \leftarrow \frac{\lVert \mY - \mX\widehat{\vbeta}\rVert^2}{n+2}. \label{sigma_update}
\end{align}
Note that this update for $\sigma^2$ is a \emph{conditional} mode, not a marginal mode, and so it does not underestimate the error variance in the same way as \eqref{conj_sigma_marginal}. Indeed, conditional on the true regression coefficients, \eqref{sigma_update} is essentially the oracle estimator \eqref{oracle_sigma}. However, although we retain the update \eqref{sigma_update} \emph{during} optimization in order to iterate between the modes of $\vbeta$ and $\sigma^2$, after the algorithm has converged, our final estimator of $\sigma^2$ is obtained as
\begin{align}
\widehat{\sigma}^2_{adj} = \frac{\lVert \mY -  \mX\widehat{\vbeta}\rVert^2}{n - \widehat{q}}, \label{sigma_adj}
\end{align}
where $\widehat{q} = \lVert \widehat{\vbeta}\rVert_0$. Note that \eqref{sigma_adj} incorporates an appropriate degrees of freedom adjustment to account for the fact that $\widehat{\vbeta}$ is an estimate of the unknown true $\vbeta$.

In principle, both $\sigma^2$ and the conditional expectation ${\theta}_{\beta}$ should be updated after each $\beta_j$, $j=1,\dots, p$.  In practice, however, there will be little change after one coordinate update and so both $\sigma^2$ and $\theta_{\beta}$ can be updated after $M$ coordinates are updated, where $M$ is the update frequency given in Algorithm 1. The default implementation updates $\sigma^2$ and $\theta_{\beta}$ after every $M = 10$ coordinate updates.

\subsection{Implementation}

In the SSL with known variance, \citet{RG18} propose a ``dynamic posterior exploration'' strategy whereby the slab parameter $\lambda_1$ is held fixed and the spike parameter $\lambda_0$ is gradually increased to approximate the ideal point mass prior. Holding the slab parameter fixed serves to stabilize the non-zero coefficients, unlike the Lasso which applies an equal level of shrinkage to all regression coefficients.  Meanwhile, gradually increasing $\lambda_0$ over a ``ladder'' of values serves to progressively threshold negligible coefficients. More practically, the dynamic strategy aids in mode detection: when $(\lambda_1 - \lambda_0)^2 \leq 4/\sigma^2$, the objective is convex \citep{RG18}. In fact, when $\lambda_0 = \lambda_1$, it is equivalent to the Lasso.  As $\lambda_0$ is increased, the posterior landscape becomes multimodal, but using the solution from the previous value of $\lambda_0$ as a ``warm start" allows the procedure to more easily find modes.  Thus, progressively increasing $\lambda_0$ acts as an annealing strategy.

When $\sigma^2$ is treated as unknown, the successive warm start strategy of \citet{RG18} will require additional intervention. For small $\lambda_0 \approx \lambda_1$, there may be many negligible but non-zero $\beta_j$ included in the model. This severe overfitting results in all the variation in $\mY$ being explained by the model, forcing the estimate of the error variance, $\widehat{\sigma}^2$ to zero.  If this suboptimal solution is propagated for larger values of $\lambda_0$, the optimization routine will remain ``stuck'' in that part of the posterior landscape.  As an implementation strategy to avoid this absorbing state, we keep the estimate of $\sigma^2$ fixed at an initial value until $\lambda_0$ reaches a value at which the algorithm converges in less than 100 iterations. We then reinitialize $\vbeta$ and $\sigma^2$ and being to simultaneously update $\sigma^2$ for the next largest $\lambda_0$ value in the ladder.  The intuition behind this strategy is that we first find a promising part of the posterior landscape and then update $\vbeta$ and $\sigma^2$. 

For initialization, we follow \citet{RG18} and initialize the regression coefficients, $\vbeta$, at zero and $\theta_0 = 0.5$.  For the error variance, we devised an initialization strategy that is motivated by the prior for $\sigma^2$ used in \citet{CGM10}.  Those authors used a scaled-inverse-$\chi^2$ prior for the error variance with degrees of freedom $\nu = 3$ and scale parameter chosen such that the sample variance of $\mY$ corresponds to the 90$th$ quantile of the prior.  This is a natural choice as the variance of $\mY$ is the maximum possible value for the error variance.  We set the initial value of $\sigma^2$ to be the mode of this scaled-inverse-$\chi^2$ distribution, a strategy which we have found to be effective in practice. 

The entire implementation strategy is summarized in Algorithm 1. 

\begin{algorithm}
\small \begin{flushleft}
Input: grid of increasing $\lambda_0$ values $I = \{\lambda_0^{1},\dots, \lambda_0^{L}\}$, update frequency $M$\\[4pt]
Initialize: $\widehat{\vbeta}_0 = \vzero_p$, $\widehat{\sigma}_0^2$, $\widehat{\theta}_0 = 0.5$\\[4pt]
For $l = 1, \dots, L$:
\begin{enumerate}
\item  Initialize: $\widehat{\vbeta}_l = \widehat{\vbeta}_{l-1}$, $\widehat{\theta}_l = \widehat{\theta}_{l-1}$, $\widehat{\sigma}_l^{2} = \widehat{\sigma}_{l-1}^2$ 
\item Set $k_l = 0$
\item While \textsf{diff} $ > \varepsilon$
\begin{enumerate}[(i)]
\item Increment $k_l$
\item For $ s= 1, \dots, \lfloor p/M \rfloor$:
\begin{enumerate}
\item Update
\begin{equation*}
\Delta \leftarrow
\begin{cases}
\sqrt{2n\widehat{\sigma}_l^2\log[1/p^*(0;\widehat{\theta}_l)]} + \widehat{\sigma}_l^2\lambda_1 &\text{if } g(0;\widehat{\theta}_l) >0\\
\widehat{\sigma}_l^2\lambda^*(0;\widehat{\theta}_l) &\text{otherwise}
\end{cases}
\end{equation*}
\item For $j = 1,\dots, M$: update $$\beta_{l(s-1)M + j} \leftarrow \widetilde{S}(z_j, \widehat{\sigma}^2\lambda^*(\beta_{l(s-1)M + j};\widehat{\theta}_l), \Delta)$$
\item Update $\widehat{\theta}_l \leftarrow (a + \lVert \widehat{\vbeta}_l \rVert_0)/( a + b +p)$
\item If $k_{l-1} < 100$:
\begin{enumerate}
\item Update $\widehat{\sigma}_l^2 \leftarrow \lVert \mY - \mX\widehat{\vbeta}_l\rVert^2/(n+2)$
\end{enumerate}
\item \textsf{diff} $= \lVert \vbeta^{k_l} - \vbeta^{k_l - 1}\rVert_2$
\end{enumerate}
\end{enumerate}
 \end{enumerate}
 \end{flushleft}
\caption{Spike-and-Slab Lasso with unknown variance}
\end{algorithm}

\subsection{Scaled Spike-and-Slab Lasso}

An alternative approach for extending the SSL for unknown variance is to follow the scaled Lasso framework of \citet{SZ12}.  In their original scaled Lasso paper, \citet{SZ12} note that their loss function can be used with many penalized likelihood procedures, including the MCP and the SCAD penalties.  Here, we develop the \emph{scaled Spike-and-Slab Lasso}.  The loss function for the scaled SSL is the same as that of the scaled Lasso but with a different penalty:
\begin{equation}
L(\vbeta, \sigma^2) = -\frac{1}{2\sigma}\lVert \mY- \mX\vbeta\rVert^2 - \frac{n\sigma}{2} + \sum_{j=1}^p pen(\beta_j|\theta_{\beta})\label{scaledSSL}
\end{equation}
where $pen(\beta_j|\theta_{\beta})$ is as defined in \eqref{ssl_pen} and again we use the approximation \eqref{theta_approx} for the conditional expectation $\theta_{\beta}$. In using this loss function, we are of course departing from the Bayesian paradigm and simply {considering} this procedure as a penalized likelihood method with a spike-and-slab inspired penalty.

The algorithm to find the modes of \eqref{scaledSSL} is very similar to Algorithm 1, the only difference being we replace all $\sigma^2$ terms in the updates \eqref{beta_update} and \eqref{delta_update} with $\sigma$.  This is because the refined thresholds for the coefficients are derived using the KKT conditions where the only difference between the two procedures is $\sigma$ vs. $\sigma^2$.

The Newton step for $\sigma^2$ is only very slightly different from the SSL with unknown variance:
\begin{equation}
\widehat{\sigma}^2 \leftarrow \frac{\lVert \mY - \mX\widehat{\vbeta}\rVert^2}{n}.
\end{equation}

How do we expect the scaled Spike-and-Slab Lasso to compare to the Spike-and-Slab Lasso with unknown variance? The threshold levels $\Delta$ for the scaled SSL will be smaller after replacing $\sigma^2$ with $\sigma$. This may potentially result in more false positives being included in the scaled SSL model. In terms of variance estimation, the updates for $\sigma^2$ are effectively the same; the only differences we should expect are those arising from a more saturated estimate for $\vbeta$. These hypotheses are examined in the simulation study in the next section.

\subsection{Simulation Study}
We now compare the Spike-and-Slab Lasso with unknown variance with several penalized likelihood methods, including the original Spike-and-Slab Lasso with fixed variance of \citet{RG18} as well as the scaled Spike-and-Slab Lasso outlined in the previous section. We investigate both the efficacy of the SSL with unknown variance and the benefits of simultaneously estimating the regression coefficients $\vbeta$ and error variance $\sigma^2$ in variable selection.  We do not consider the SSL with the $p$-sigma prior as the objective is similar to \citet{S10} (albeit with an adaptive penalty) and so we would expect similar overestimation of $\sigma^2$ as proved by \citet{SZ12}. We consider three different simulation settings.

For the first simulation setting, we consider the same simulation setting of \citet{RG18} with $n=100$ and $p=1000$ but use an error variance of $\sigma^2 = 3$ instead of $\sigma^2 =1$. The data matrix $\mX$ is generated from a multivariate Gaussian distribution with mean $\vzero_p$ and a block-diagonal covariance matrix $\mSigma = \text{bdiag}(\widetilde{\Sigma}, \dots, \widetilde{\Sigma})$ where $\widetilde{\Sigma} = \{\widetilde{\sigma}\}_{i,j=1}^{50}$ with $\widetilde{\sigma}_{ij} = 0.9$ if $i\neq j$ and $\widetilde{\sigma}_{ii} = 1$.  The true vector $\vbeta_0$ is constructed by assigning regression coefficients $\{-2.5, -2, -1.5, 1.5, 2, 2.5\}$ to $q=6$ entries located at $\{1, 51, 101, 151, 201, 251\}$ and setting to zero the remaining coefficients. Hence, there are 20 independent blocks of 50 highly correlated predictors where the first 6 blocks each contain only one active predictor. The response was generated as in \eqref{regression} with error variance $\sigma^2 = 3$.

We compared the Spike-and-Slab Lasso with unknown variance to the fixed variance Spike-and-Slab Lasso with two settings: (i) $\sigma^2 =1$, and (ii) $\sigma^2 = 3$, the true variance. The prior settings for $\theta$ were $a = 1, b = p$.  The slab parameter was set to $\lambda_1 = 1$.  For the spike parameter, we used a ladder $\lambda_0 \in I = \{1, 2, \dots, 100\}$. 

Additional methods compared were the scaled SSL from Section 3.4, the Lasso \citep{F10}, the scaled Lasso \citep{SZ12}, the Adaptive Lasso \citep{Z06}, SCAD \citep{FL01}, and MCP \citep{Z10}.

The analysis was repeated 100 times with new covariates and responses generated each time.  For each, the metrics recorded were: the Hamming distance (HAM) between the support of the estimated $\vbeta$ and the true $\vbeta_0$; the prediction error (PE), defined as
\begin{align}
\text{PE} = \lVert \mX\vbeta_0 - \mX\widehat{\vbeta}\rVert^2;
\end{align}
the number of false negatives (FN); the number of false positives (FP); the number of true positives (TP); Matthew's Correlation Coefficient (MCC), defined as:
\begin{align}
\text{MCC} = \frac{TP \times TN - FP \times FN}{\sqrt{(TP + FP)(TP + FN)(TN + FP)(TN + FN)}};
\end{align}
the dimension of the estimated $\vbeta$ ($\widehat{q}$); the percentage of times the method found the correct model  (COR); and the time in seconds (TIME).  The average of these metrics for each method over the 100 repetitions are displayed in Table \ref{sim_study1}.  

\begin{table}[t]
\centering \small
\hspace*{-2cm}\begin{tabular}{lllllllll}
  \hline\hline \\[-0.3cm]
 & HAM &PE&MCC &TP &FP &FN  &COR &TIME \\ 
  \hline \\[-0.3cm]
SSL (fixed $\sigma^2 = 3$) & \bf1.1 {\scriptsize (0.1)} & \bf39.6 {\scriptsize (3.7)} &\bf 0.91 {\scriptsize (0.01)} &\bf 5.5 {\scriptsize (0.1)} &\bf 0.5 {\scriptsize (0.1)} &\bf 0.5 {\scriptsize (0.1)} &\bf58 &\bf 0.03 {\scriptsize (0.00)} \\ 
  SSL (unknown $\sigma^2$) & 1.2 {\scriptsize (0.2)} & 43.4 {\scriptsize (3.9)} & 0.90 {\scriptsize (0.01)} & 5.4 {\scriptsize (0.1)} & 0.6 {\scriptsize (0.1)} & 0.6 {\scriptsize (0.1)} & 55 & 0.04 {\scriptsize (0.00)} \\ 
  Scaled SSL & 2.0 {\scriptsize (0.2)} & 65.8 {\scriptsize (5.0)} & 0.84 {\scriptsize (0.01)} & 5.2 {\scriptsize (0.1)} & 1.2 {\scriptsize (0.1)} & 0.8 {\scriptsize (0.1)} & 32 & 0.07 {\scriptsize (0.00)} \\ 
  SSL (fixed $\sigma^2 = 1$) & 4.5 {\scriptsize (0.3)} & 114.9 {\scriptsize (5.3)} & 0.69 {\scriptsize (0.02)} & 4.8 {\scriptsize (0.1)} & 3.3 {\scriptsize (0.2)} & 1.2 {\scriptsize (0.1)} & 5 & 0.17 {\scriptsize (0.01)} \\ 
  MCP** & 7.0 {\scriptsize (0.4)} & 186.1 {\scriptsize (7.0)} & 0.48 {\scriptsize (0.02)} & 3.1 {\scriptsize (0.1)} & 4.1 {\scriptsize (0.3)} & 2.9 {\scriptsize (0.1)} & 1 & 0.32 {\scriptsize (0.00)} \\ 
  Adaptive LASSO & 8.1 {\scriptsize (0.5)} & 92.0 {\scriptsize (4.1)} & 0.60 {\scriptsize (0.02)} & 4.8 {\scriptsize (0.1)} & 6.9 {\scriptsize (0.5)} & 1.2 {\scriptsize (0.1)} & 1 & 5.36 {\scriptsize (0.11)} \\ 
  SCAD & 11.2 {\scriptsize (0.6)} & 124.4 {\scriptsize (6.2)} & 0.47 {\scriptsize (0.02)} & 4.0 {\scriptsize (0.1)} & 9.2 {\scriptsize (0.5)} & 2.0 {\scriptsize (0.1)} & 0 & 0.39 {\scriptsize (0.01)} \\ 
  MCP* & 11.5 {\scriptsize (0.4)} & 181.4 {\scriptsize (6.3)} & 0.35 {\scriptsize (0.02)} & 2.8 {\scriptsize (0.1)} & 8.3 {\scriptsize (0.3)} & 3.2 {\scriptsize (0.1)} & 0 & 0.32 {\scriptsize (0.00)} \\ 
  Scaled LASSO & 16.1 {\scriptsize (0.4)} & 302.4 {\scriptsize (9.6)} & 0.42 {\scriptsize (0.01)} & 4.5 {\scriptsize (0.1)} & 14.6 {\scriptsize (0.4)} & 1.5 {\scriptsize (0.1)} & 0 & 0.51 {\scriptsize (0.01)} \\ 
  LASSO & 30.9 {\scriptsize (0.6)} & 111.0 {\scriptsize (2.5)} & 0.36 {\scriptsize (0.01)} & 5.4 {\scriptsize (0.1)} & 30.3 {\scriptsize (0.6)} & 0.6 {\scriptsize (0.1)} & 0 & 0.40 {\scriptsize (0.01)} \\ \hline\hline
\end{tabular}\hspace*{-2cm}
\caption{Average metrics over 100 repetitions for each of the procedures, ordered by increasing Hamming distance. Standard errors are shown in parentheses. *\texttt{ncvreg} implementation using cross-validation over a one-dimensional grid with a default value of the second tuning parameter $\gamma$.  **hard thresholding tuning with $\gamma = 1.0001$ and cross-validation over $\lambda$.}\label{sim_study1}
\end{table}

We can see that the Spike-and-Slab Lasso with the variance fixed and equal to the truth ($\sigma^2 = 3$) performs the best in terms of the Hamming distance, prediction error, and MCC.  Encouragingly, the Spike-and-Slab Lasso with unknown variance performs almost as well as the ``oracle'' version where the true variance is known. The SSL with unknown variance in turn performs better than a naive implementation of the SSL with fixed variance ($\sigma^2 =1$). We note that the prediction error for the latter implementation is higher than the Adaptive Lasso and SCAD; however, these frequentist methods use cross-validation to choose their regularization parameter and so are optimizing for prediction to the possible detriment of other metrics; the SSL ($\sigma^2=1$) still has fewer false positives and a higher MCC.  However, both the SSL ($\sigma^2=3$) and unknown $\sigma^2$ have smaller prediction error than the rest of the methods, including those which use cross-validation, which highlights the predictive gains afforded by variance estimation. 

Following from the discussion in Section 3.4, we can see that the scaled SSL indeed finds more false positives than the SSL with unknown variance. This is a result of the smaller thresholds in estimating the regression coefficients. We can see that the scaled Lasso significantly reduces the number of false positives found as compared to the Lasso; however, the issues with the Lasso penalty remain.

\begin{figure}
\centering
\includegraphics[width = 0.6\textwidth]{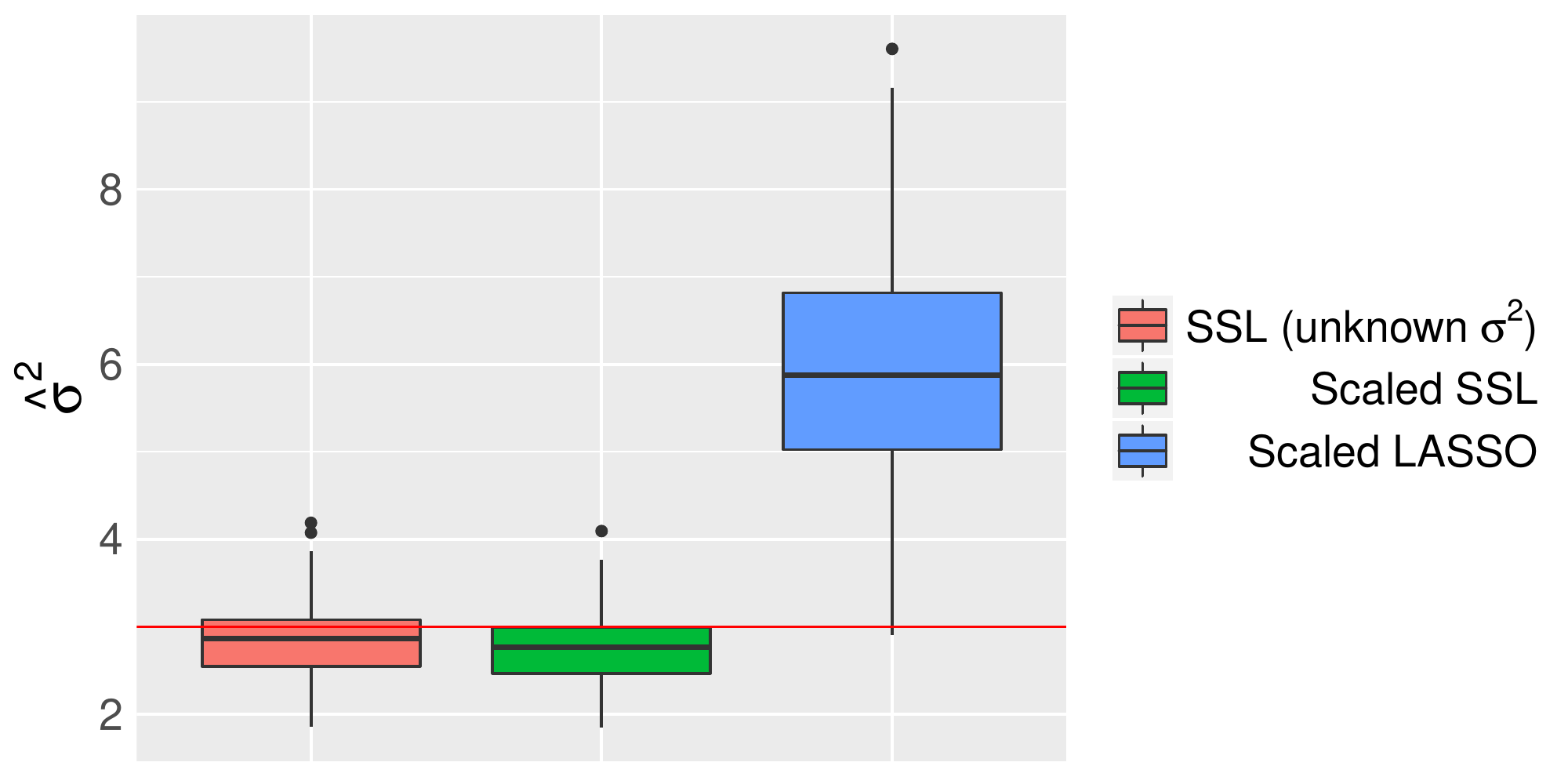}
\caption{Estimated $\widehat{\sigma}_{adj}^2$ over 100 repetitions. The true variance $\sigma^2 = 3$ is the red horizontal line.}
\label{sigmas}
\end{figure}

Figure \ref{sigmas} shows the variance estimates over the 100 repetitions for the SSL with unknown variance, the scaled SSL and the scaled Lasso. For the SSL with unknown $\sigma^2$, these are the estimates \eqref{sigma_adj}. For the scaled SSL and the scaled Lasso variance estimates, we also applied a degrees of freedom correction similarly to \eqref{sigma_adj} using the number of non-zero coefficients found by each method. The variance estimates from the SSL (unknown $\sigma^2$) have a median of 2.87 and standard error 0.04.  Meanwhile, the scaled SSL slightly underestimates the variance with a median of 2.76 and standard error 0.04, as expected from the larger number of false positives observed in Table \ref{sim_study1}. Finally, the scaled Lasso highly inflates the variance with a median of 5.88 and standard error 0.14.

\section{Protein activity data}

We now apply the Spike-and-Slab Lasso with unknown variance to the protein activity data set from \citet{CP98}.  Following those authors, we code the categorical variables by indicator variables, consider all main effects and two-way interactions, and quadratic terms for the continuous variables. This results in a linear model with $p = 88$ potential predictors. The sample size is $n = 96$.  We assess the performance of the Spike-and-Slab Lasso with unknown variance in both variable selection and prediction.

\subsection{Variable selection}
As an approximation to the ``truth'', we use the Bayesian adaptive sampling algorithm  \citep[BAS,][]{C11}, which has previously been applied successfully to this dataset.  BAS gives posterior inclusion probabilities (PIP) for each of the $p$ potential predictors from which we determined the median probability model (MPM: predictors with PIP $ > 0.5$). The median probability model found by BAS consisted of $q = 7$ predictors: 
\begin{itemize}
\item interaction of protein concentration and detergent T (\texttt{con:detT})
\item detergent T (\texttt{detT})
\item interaction of  buffer TRS and detergent N (\texttt{bufTRS:detN})
\item protein concentration (\texttt{con})
\item interaction of buffer P04 and temperature (\texttt{bufPO4:temp})
\item detergent N (\texttt{detN})
\item interaction of detergent N and temperature (\texttt{detN:temp})
\end{itemize}

For the SSL with unknown variance, we used the same settings as the simulation study with $\lambda_1 = 1$ and $\lambda_0 \in \{1,2,\dots, n\}$.  The procedure found a model with $\widehat{q} = 6$ predictors, including four of the MPM: \texttt{con, detN, bufTRS:detN, con:detT}. Additionally, instead of \texttt{detT}, the SSL with unknown variance found the interaction of pH with detergent T (\texttt{pH:detT}). The correlation between \texttt{detT} and \texttt{pH:detT} is 0.988, rendering the two predictors essentially exchangeable. Thus, 5 out of the 6 predictors found by the SSL with unknown variance matched with the benchmark MPM. 

For the SSL with known variance, we fixed $\sigma^2 = 0.24$. This is the mean of the scaled-inverse-$\chi^2$ distribution induced by the variance of the response, as detailed in Section 6.4. For the protein data, the sample variance of the response is 0.41 and so fixing $\sigma^2 = 1$ overestimates the variance, resulting in no signal being found.   The SSL with this fixed variance found $\widehat{q} = 2$ predictors: one of the MPM (\texttt{detT}) and one not in the MPM but having a correlation of 0.735 with \texttt{detN}.  

Here, we can see the benefit of simultaneously estimating the error variance; the estimate from SSL with unknown variance was $\widehat{\sigma}^2 = 0.167$, resulting in a less sparse solution.

\subsection{Predictive Performance}

We now compare the predictive performance of the SSL with unknown variance with the penalized regression methods from the simulation study in Section 6.5 using 8-fold cross validation. We split the data into $K = 8$ sets and denote each set by $S_k$, $k = 1,\dots, K$. The 8-fold cross-validation error is given by:
\begin{equation}
CV = \frac{1}{K}\sum_{k = 1}^K\sum_{i \in S_k} \left[y_i -\vx_i\widehat{\vbeta}_{\backslash k}\right]^2
\end{equation}
where $\widehat{\vbeta}_{\backslash k}$ is the estimated regression coefficient using the data in $S_k^C$.   We repeated this procedure 100 times and display the resulting cross-validation errors in Figure \ref{protein_boxplot}.  We do not display the results from the scaled Lasso in Figure \ref{protein_boxplot} as there were a number of outliers: the cross-validation error for the scaled Lasso was greater than 25 in 20\% of the replications.

\begin{figure}
\centering
\includegraphics[width = 0.6\textwidth]{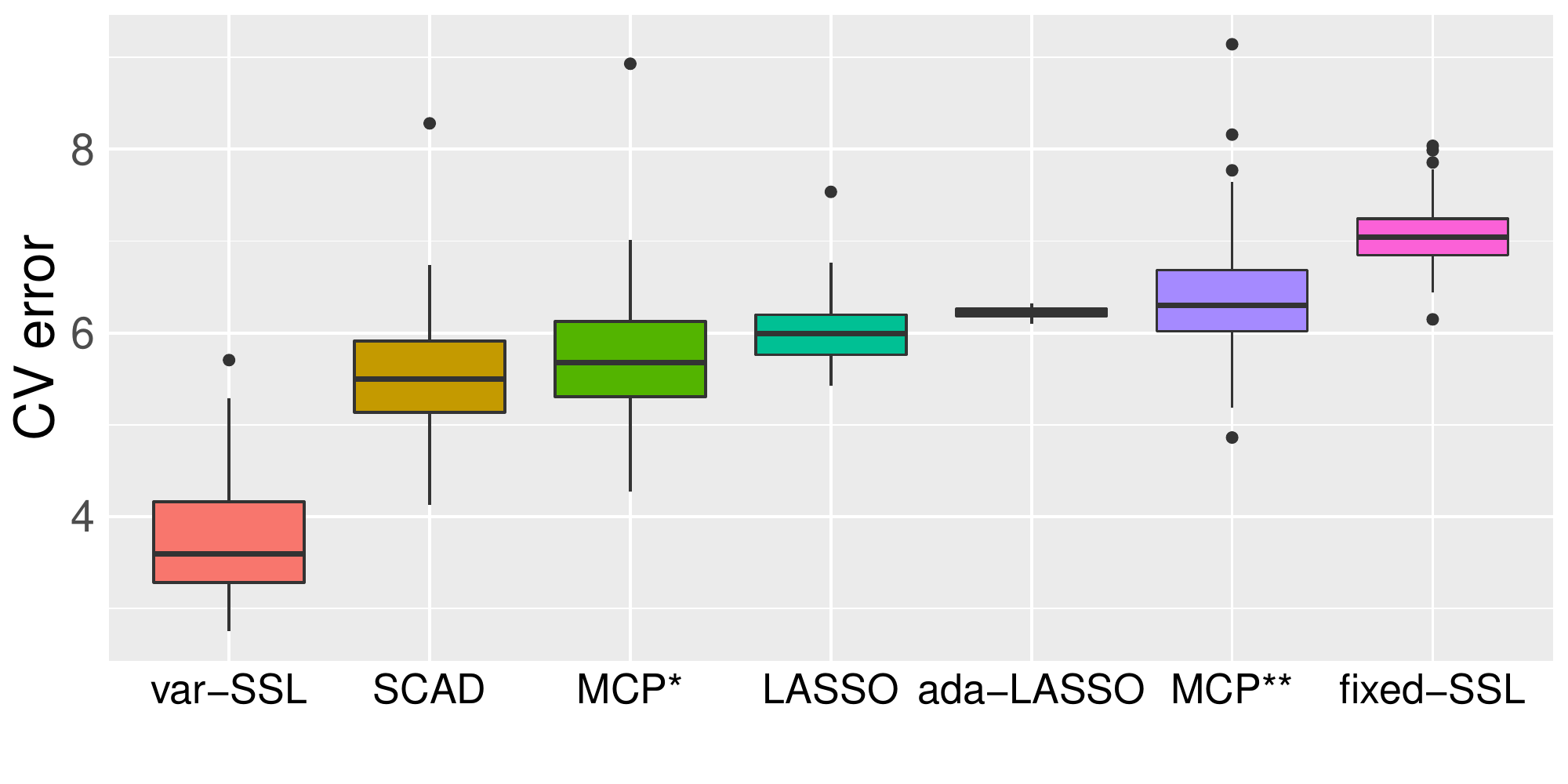}
\caption{Boxplots of 8-fold cross-validation error over 100 replications for each of the methods (from left to right): 1. SSL (unknown $\sigma^2$). 2. SCAD. 3. MCP (\texttt{ncvreg}). 4. LASSO. 5. Adaptive LASSO. 6. MCP ($\gamma = 1.0001$). 7. SSL (fixed $\sigma^2$).}
\label{protein_boxplot}
\end{figure}

We can see that the SSL with unknown variance has the smallest cross-validation error.  This highlights the gains in predictive performance that can be achieved by simultaneously estimating the error variance and regression coefficients. This result is also very encouraging given that all the other methods (except for the SSL with fixed variance) use cross-validation in choosing their regularization parameters.  This also explains the slightly worse performance of the SSL with fixed variance, which we expect would be competitive if we were to also choose the regularization parameters with cross-validation.  However, SSL with fixed variance still performs well without the need for computationally intensive cross-validation to choose the parameters.

\section{Conclusion}

In this paper, we have shown that conjugate continuous priors for Bayesian variable selection can lead to underestimation of the error variance when (i) $\vbeta$ is sparse; and (ii) when $p$ is of the same order as, or larger than, $n$. This is because such priors implicitly add $p$ ``pseudo-observations'' of $\sigma^2$ which shift prior mass on $\sigma^2$ towards zero. Conjugate priors for linear regression are often motivated by the invariance principle of \citet{J61}. Revisiting this work however, we highlighted that Jeffreys' himself cautioned against applying his invariance principle in multivariate problems.  Following Jeffreys. we recommended priors which treat the regression coefficients and error variance as independent.  

We then proceeded to extend the Spike-and-Slab Lasso of \citet{RG18} to the unknown variance case, using an independent prior for the variance. We showed that this procedure for the Spike-and-Slab Lasso with unknown variance performs almost as well empirically as the SSL where the true variance is known. We additionally compared the Spike-and-Slab Lasso with unknown variance to a popular frequentist method to estimate the variance in high dimensional regression: the scaled Lasso.  In simulation studies, the SSL with unknown variance performed much better than the scaled Lasso and additionally outperformed the ``scaled Spike-and-Slab Lasso'', a variant of the latter procedure but with the Spike-and-Slab Lasso penalty.  On a protein activity dataset, the SSL with unknown variance performed well for both variable selection and prediction.  In particular, the SSL with unknown variance exhibited smaller cross-validation error than other penalized likelihood procedures which choose their regularization parameters based on cross-validation. This highlights the predictive benefit of simultaneous variance estimation.  The unknown variance implementation of the SSL is provided in the publicly available R package \texttt{SSLASSO} \citep{SSLpackage}. Code to reproduce the results in this paper is also available at \url{https://github.com/gemma-e-moran/variance-priors}.

\bibliographystyle{ba}
\bibliography{references} 

\subsubsection*{Acknowledgements}
This research was supported by the NSF Grant DMS-1406563 and the James S. Kemper Foundation Faculty Research Fund at the University of Chicago Booth School of Business. We would like to thank the Editors, Associate Editor, and anonymous referees, as well as Mark van de Wiel and Gwenael Leday, for helpful suggestions which improved this paper. 

\section*{Appendix}
\begin{appendix}

\section{Gibbs sampler for Bayesian ridge regression}

We derive the Gibbs sampler used to obtain posterior estimates for the independent Bayesian ridge regression model in Section 3.2 of the main paper. The model is:
\begin{align}
&\mY \sim N_n(\mX\vbeta, \sigma^2\mI) \\
&\vbeta \sim N_p(0, \tau^2\mI) \\
&\pi(\sigma) \propto 1/\sigma.
\end{align}

The full conditional distributions of the parameters $\vbeta$ and $\sigma^2$ are:
\begin{align}
\vbeta|\mY, \sigma^2 &\sim N_p(\sigma^{-2}\mV\mX^T\mY, \mV) \label{gibbs_beta}\\
\sigma^2|\mY, \vbeta &\sim IG(n/2, \lVert \mY - \mX\vbeta\rVert^2/2)\label{gibbs_sigma}
\end{align}
where $\mV = \left[\sigma^{-2}\mX^T\mX + \tau^{-2}\mI_p\right]^{-1}$. The Gibbs sampling algorithm alternates sampling from \eqref{gibbs_beta} and \eqref{gibbs_sigma}. After burn-in, the posterior mean estimate is the mean of the samples.

\section{EM Algorithm for the Spike-and-Slab lasso with a conjugate prior}

We provide the details of the EM algorithm for the Spike-and-Slab Lasso with a conjugate prior in Section 6.2 of the main paper.  The model is given by:
\begin{align}
\pi(\vbeta|\vgamma, \sigma^2) &\sim \prod_{j=1}^p \left(\gamma_j\frac{\lambda_1}{2\sigma}e^{-|\beta_j|\lambda_1/\sigma} + (1-\gamma_j)\frac{\lambda_0}{2\sigma}e^{-|\beta_j|\lambda_0/\sigma}\right)\\
\vgamma|\theta &\sim \prod_{j=1}^p \theta^{\gamma_j}(1-\theta)^{1-\gamma_j},\quad \theta\sim \text{Beta}(a,b)\\
p(\sigma^2) &\propto \sigma^{-2}. 
\end{align}
Then, the ``complete'' data log posterior is given by 
\begin{align}
\log \pi(\vbeta, \vgamma, \sigma, \theta|\mY) &= -\frac{1}{2\sigma^2}\lVert \mY - \mX\vbeta\rVert^2 - (n + 2)\log\sigma \notag\\
&\quad + \sum_{j=1}^p \log \left(\gamma_j \frac{\lambda_1}{2\sigma}e^{-|\beta_j|\lambda_1/\sigma} + (1-\gamma_j)\frac{\lambda_0}{2\sigma}e^{-|\beta_j|\lambda_0/\sigma}\right)\notag  \\
&\quad + \sum_{j=1}^p \log\left(\frac{\theta}{1-\theta} \right)\gamma_j  + (a-1)\log(\theta) \notag \\
&\quad+ (p+b-1)\log(1-\theta) + C \label{complete_v1}
\end{align}
The EM algorithm then proceeds as follows: treat $\vgamma$ as unknown and iteratively maximize 
\begin{equation}
E[\log \pi(\vbeta,\vgamma, \sigma, \theta|\mY) |\vbeta^{(k)}, \sigma^{(k)}, \theta^{(k)}]
\end{equation}
where $\vbeta^{(k)}, \sigma^{(k)}, \theta^{(k)}$ are the parameter values after the $k$th iteration.

\end{appendix}

\end{document}